\documentclass[aps,amssymb,amsfonts,twocolumn,pra]{revtex4}
\usepackage[latin1]{inputenc}
\usepackage{amsthm}
\usepackage{amsmath}
\usepackage{graphicx} 
\usepackage{subfigure}
\usepackage{amsfonts}
\usepackage{hyperref} 
\usepackage{xcolor}
\usepackage{verbatim}
\usepackage{listings}
\usepackage{hyperref}
\usepackage{soul}
\usepackage{bbm}
\usepackage{mathtools}
\usepackage{tikz}
\usetikzlibrary{matrix}
\usepackage{cases}
\usepackage{MnSymbol}

\newtheorem{defi}{Definition}
\newtheorem{prop}{Proposition}
\newtheorem{conje}{Conjecture}
\newtheorem{obs}{Observation}
\newcommand{\ket}[1]{\left\vert #1\right\rangle}

\renewcommand{\eqref}[1]{Eq.~(\ref{#1})}
\newcommand{\floor}[1]{\left\lfloor #1 \right\rfloor}

\def\e{\mathrm{e}}

\DeclareMathOperator{\hiH}{\mathcal{H}}
\newcommand{\1}{\ensuremath{\mathbbm{1}}}
\newcommand\LS{\mathrm{LS}}
\newcommand\QLS{\mathrm{QLS}}
\newcommand\LC{\mathrm{LC}}

\newcommand\OA{\mathrm{OA}}
\newcommand\QOA{\mathrm{QOA}}
\newcommand\AME{\mathrm{AME}}
\newcommand\IrOA{\mathrm{IrOA}}

\begin{document}
\title{Entanglement and quantum combinatorial designs}
\author{Dardo Goyeneche}
\affiliation{Institute of Physics, Jagiellonian University, 30-348 Krak\'ow, Poland}
\affiliation{Faculty of Applied Physics and Mathematics, Technical University of Gda\'{n}sk, 80-233 Gda\'{n}sk, Poland}
\affiliation{Departamento de F\'{i}sica, Facultad de Ciencias B\'{a}sicas, Universidad de Antofagasta, Casilla 170, Antofagasta, Chile}
\author{Zahra Raissi}
\affiliation{ICFO-Institut de Ciencies Fotoniques, The Barcelona Institute of Science and Technology, Castelldefels (Barcelona), 08860, Spain}
\author{Sara Di Martino}
\affiliation{F\'isica Te\`orica: Informaci\'o i Fen\`omens Qu\`antics, Departament de F\'isica, Universitat Aut\`onoma de Barcelona, 08193 Bellaterra (Barcelona), Spain}
\author{Karol {\.Z}yczkowski}
\affiliation{Institute of Physics, Jagiellonian University, 30-348 Krak\'ow, Poland}
\affiliation{Center for Theoretical Physics, Polish Academy of Sciences, 02-668 Warsaw, Poland}

\date{April 9, 2018}

\begin{abstract}
We introduce several classes of quantum combinatorial designs, namely quantum Latin squares, 
cubes, hypercubes and a notion of orthogonality between them. A further introduced notion,
 quantum orthogonal arrays, generalizes all previous classes of designs. 
We show that mutually orthogonal quantum Latin arrangements can
 be entangled in the same way than quantum states are entangled. 
Furthermore, we show that such designs naturally define a remarkable class of genuinely multipartite highly entangled states called $k$-uniform, i.e. multipartite pure states such that every reduction to $k$ parties is 
maximally mixed. We derive infinitely many classes of mutually orthogonal quantum Latin 
arrangements and quantum orthogonal arrays having an arbitrary large number of columns. 
The corresponding multipartite $k$-uniform states exhibit a high persistency of entanglement, which makes them ideal candidates to develop multipartite quantum information protocols.
\end{abstract}

\maketitle
Keywords: combinatorial designs, quantum orthogonal arrays, $k$-uniform states.

\section{Introduction}
One of the key problems in the theory of quantum information is
 to identify multipartite quantum states with the strongest
 possible quantum correlations.
Contrary to the classical behaviour, information stored in multipartite quantum systems is not equivalent to information provided by the parties. The extremal situation occurs when information stored in an $N$ partite pure quantum state is not present \emph{at all} in any subset of $k$ collaborating parties, for some integer $k\leq N/2$. 
Such pure states are called \emph{k-uniform}~\cite{S04,AC13,GZ14}, 
meaning that every reduction to $k$ parties is given by the maximally mixed state. When $k=\floor{N/2}$, where $k=\floor{ .}$ denotes the floor function,
the state is called \emph{Absolutely Maximally Entangled} (AME).
Sometimes, these states are also called Maximally Multi-partite Entangled States 
 \cite{FFPP08}, or MMES for short. 

For instance, the generalized Bell states of two subsystems with $d$ levels each
and the tripartite GHZ-like states
belong to the AME class. These highly entangled states 
 find applications in quantum secret sharing \cite{HCLRL12}, 
quantum error correction codes \cite{S04} and holographic codes \cite{PYHP15}.
 They can be constructed from graph states \cite{HEB04,H13}, orthogonal arrays \cite{GZ14},
 multi-unitary matrices \cite{GALRZ15} and perfect tensors \cite{PYHP15,S04}. 
Furthermore, from gluing AME states further multipartite classes of such states can be constructed in higher dimensions \cite{RK17}. However, to determine the existence of $\AME(N,d)$ for any number of parties $N$ and internal levels $d$ is a difficult problem, specially if $d$ is not a power of a prime number \cite{Hu17}.
 Many approaches were tried in order to give an answer to this question, including 
recasting of the problem in language of statistical mechanics
\cite{cumulants,Raissi-2017,cactus}.

In this work, we introduce certain classes of combinatorial designs by extending classical symbols to pure quantum states. Our starting point is the notion of \emph{quantum Latin squares} (QLS) \cite{MV15}, which we generalize to quantum Latin cubes (QLC) and hypercubes (QLH).
 We also introduce a notion of orthogonality between them and identify
 a crucial  ingredient missing in the previous approach \cite{M17}: two orthogonal QLS could be entangled, in such a way that they cannot be expressed as two separated arrangements. These entangled designs are intrinsically associated to a larger class of quantum designs that includes all previous quantum Latin arrangements: \emph{quantum orthogonal arrays}. After setting up the quantum combinatorial tools we apply our method to the problem to construct $k$-uniform states and absolutely maximally entangled states in particular, for multipartite systems having an arbitrary large number of parties.

The paper is organized as follows: In Section \ref{S2}, we recall the 
standard concepts of (classical) Latin squares, Latin cubes, Latin hypercubes and orthogonal arrays and review their basic properties. In Section \ref{S3} we define quantum Latin squares, cubes, hypercubes and introduce a notion of orthogonality between them. Simple examples in low dimensions are provided. In Section \ref{S4}, we introduce the concept of quantum orthogonal arrays. We show that quantum Latin arrangements arise from quantum orthogonal arrays in the same way that Latin arrangements arise from orthogonal arrays in combinatorics. 
In Section \ref{S5} we show a connection existing between quantum orthogonal arrays and multiunitary matrices, the last ones introduced in Ref. \cite{GALRZ15}. In Section \ref{S6} we derive simple constructions of $k$-uniform and AME states from quantum orthogonal arrays. A summary of results and concluding remarks are presented in Section \ref{S7}.

\section{Latin arrangements and orthogonal arrays}\label{S2}
In this section, we review some basic combinatorial concepts used in this work. 
A Latin square $\LS(d)$ is a square arrangement of size $d$ such that every entry, taken from the set $\{0,\dots,d-1\}$, occurs once in each row and each column. For instance, arrangements
\begin{equation}
\begin{array}{cc}
0&1\\
1&0\\
\end{array}, \hspace{0.6cm}
\begin{array}{ccc}
0&1&2\\
2&0&1\\
1&2&0
\end{array}, \hspace{0.6cm}
\begin{array}{cccc}
0&1&2&3\\
1&0&3&2\\
2&3&0&1\\
3&2&1&0
\end{array},
\end{equation}
are Latin squares of size $d$ equal to two, three and four, respectively.

An \emph{orthogonal array}, denoted as $\OA(r,N,d,k)$, is an arrangement composed by $r$ rows, $N$ columns and entries taken from the set $\{0,\dots,d-1\}$, such that every subset of $k$ columns contains all possible combinations of symbols, occurring the same number ($\lambda$) of times along the rows. Here, parameters $k$ and $\lambda$ are called \emph{strength} and \emph{index} of the OA, respectively \cite{HSS99}. An OA is called \emph{irredundant} if every subset of ($N-k$) columns contains no repeated rows \cite{GZ14}. Two OA are called equivalent if one array can be transformed into the other one by applying permutations or relabelling of symbols in rows or columns.

It is simple to show that any $\LS(d)$ is equivalent to an $\OA(d^2,3,d,2)$ 
-- see Chapter 8 in Ref. \cite{HSS99}.  For example, the array 
$\OA(4,3,2,2)$ produces a $\LS(2)$, as shown below:
 \begin{equation}\label{OA4322}
OA=\begin{array}{ccc}
0&0&0\\
0&1&1\\
1&0&1\\
1&1&0 \\
-- & -- & -- \\
i & j & LS
\end{array}\hspace{0.5cm}\Rightarrow\hspace{0.5cm}
LS=\begin{array}{cc}
0&1\\
1&0
\end{array}.
\end{equation}
Here, the first two columns of the OA identify coordinates  $(i,j)$
of symbols for the LS, whose values are determined by the third column $LS$
of the OA.

Two Latin squares $\LS^A$ and $\LS^B$ of size $d$ are \emph{orthogonal} if the set of ordered pairs $[(\LS^A)_{ij}, (\LS^B)_{ij}]$ is composed by all possible $d^2$ combinations symbols, where $i,j \in \{0,\dots,d-1 \}$. A collection of $m$ LS of order $d$ is called mutually orthogonal (MOLS) if they are pairwise orthogonal. For instance, any $\OA(d^2,2+m,d,2)$ defines a set of $m$ MOLS of size $d$ \cite{HSS99}. In particular, an $\OA(9,4,3,2)$ implies two classical OLS of size 3. As before, first two columns $(i,j)$ of the OA address entries of OLS,
 while the two latter yield the values of the squares A and B,
\begin{equation}\label{MOLS3}
  \begin{split}
OA(9,4,3,2)=\begin{array}{cccc}
0&0&0&0\\
0&1&2&1\\
1&0&2&2\\
1&1&1&0\\
1&2&0&1\\
2&1&0&2\\
2&2&2&0\\
2&0&1&1\\
0&2&1&2\\
- & - & - & -\\
i  & j  & A  & B  
\end{array}  \end{split}
\quad\Rightarrow\quad
  \begin{split}
\LS^A=
\begin{array}{ccc}
0&2&1\\
2&1&0\\
1&0&2
\end{array}\\[0.5cm]
\LS^B=
\begin{array}{ccc}
0&1&2\\
2&0&1\\
1&2&0
\end{array} .
  \end{split}
\end{equation}
Entries of two OLS are typically dnoted as ordered pairs in a single array. For instance, the two OLS of Eq.(\ref{MOLS3}) are denoted as 
\begin{equation}
OLS=\begin{array}{ccc}
00&21&12\\
22&10&01\\
11&02&20
\end{array}.
\end{equation}
Furthermore, orthogonal arrays can be associated to Latin cubes. An $\OA(d^3,4,d,3)$ defines a Latin cube $\LC(d)$, which consists on a cubic arrangement composed by $d$ rows, $d$ columns and $d$ files, such that every entry taken from the set $\{0,\dots,d-1\}$ occurs once in each row, each column and each file. For instance, $\OA(8,4,2,3)$ defines a LC of size 2, where now first three bits  $(i,j,k)$ determine the position of a given element of the cube $LC$, while the last bit determines its value,
\begin{equation}\label{OA8422}
OA(8,4,2,3)=\begin{array}{cccc}
0&0&0&0\\
0&0&1&1\\
0&1&0&1\\
0&1&1&0\\
1&0&0&1\\
1&0&1&0\\
1&1&0&0\\
1&1&1&1 \\
- & - & - & -\\
i  & j  & k  & LC 
\end{array} \, , \
LC=\!\begin{tikzpicture}
[
baseline={([yshift=-.5ex]current bounding box.center)},
back line/.style={dashed},
cross line/.style={preaction={draw=white, -,
line width=6pt}}]
\matrix (m) [matrix of math nodes,
row sep=0.8em, column sep=1em,
text height=1.5ex,
text depth=0.25ex]
{
& \mathbf{1} & & \mathbf{0} \\
\mathbf{0} & & \mathbf{1} \\
& \mathbf{0} & & \mathbf{1} \\
\mathbf{1} & & \mathbf{0} \\
};
\path[-]
(m-1-2) edge [back line] (m-1-4)
edge [back line] (m-2-1)
edge [back line] (m-3-2)
(m-1-4) edge [back line] (m-3-4)
edge [back line] (m-2-3)
(m-2-1) edge [back line] (m-2-3)
edge [back line] (m-4-1)
(m-3-2) edge [back line] (m-3-4)
edge [back line] (m-4-1)
(m-4-1) edge [back line] (m-4-3)
(m-3-4) edge [back line] (m-4-3)
(m-2-3) edge [back line] (m-4-3);
\end{tikzpicture}
\end{equation}
In general, an $\OA(d^k,k+m,d,k)$ defines $m$ mutually orthogonal Latin hypercube (LH) of size $d$ in dimension $k$, denoted MOLH($d$). Figure \ref{Fig1} summarizes existing relations between OA and Latin arrangements.

To emphasize the difference between the above described standard combinatorial designs and
their quantum generalizations discussed in subsequent sections we will
 refer to OA, LS and MOLS and MOLC as the \emph{classical} arrangements.
\begin{figure}[!h]
\centering 
{\includegraphics[width=6cm]{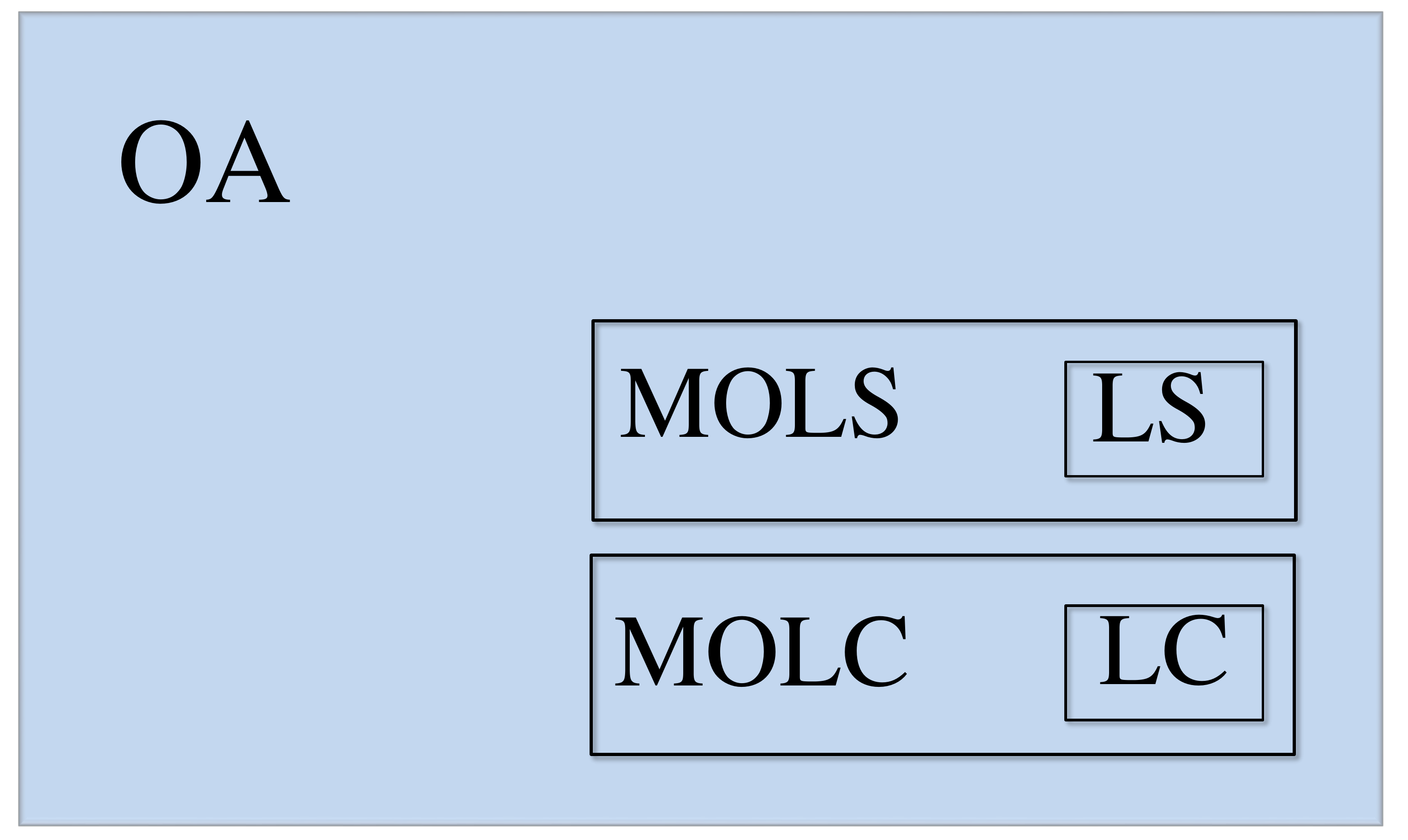}} 
\caption{Orthogonal arrays generalize some classes of combinatorial arrangements: Latin squares (LS), Latin cubes (LC), and mutually orthogonal LS and LC (MOLS and MOLC, respectively). These arrangements can be generalized to Latin hypercubes (LH) and mutually orthogonal LH (MOLS), respectively. 
Along this work, we develop a theory of quantum combinatorial designs and show that quantum Latin arrangements arise from QOA in the same way than classical Latin arrangements arise from OA.}
\label{Fig1}
\end{figure}
 An OA having $r$ rows, $N$ columns and $d$ symbols can be associated with a pure quantum state of $N$ qudit system having $r$ terms \cite{GZ14}. Each row of the array corresponds to a single term of the state, so left hand side of the arrangement  (\ref{MOLS3}) yields the unnormalized state of four parties
\begin{eqnarray}\label{AME43}
|\phi_{4,3}\rangle&=&|0000\rangle+|0121\rangle+|1022\rangle+\nonumber\\
&&|1110\rangle+|1201\rangle+|2102\rangle+\nonumber\\
&&|2220\rangle+|2011\rangle+|0212\rangle.
\end{eqnarray}
This state is maximally entangled with respect to the $\binom{4}{2}=6$ possible balanced bipartitions and it is called \emph{absolutely maximally entangled state}, denoted $\AME(4,3)$ \cite{Popescu}. Along the work we consider unnormalized pure states, for the sake of simplicity.
\section{Quantum Latin arrangements} \label{S3}
Recently, quantum Latin squares (QLS) \cite{MV15} and \emph{weakly orthogonal} QLS \cite{M17} have been introduced, where classical symbols appearing in entries of arrangements were extended to quantum states. These concepts were used to define unitary error bases \cite{MV15} and mutually unbiased bases \cite{M17}. In this section, we extend those results by introducing some classes of quantum Latin arrangements, where QLS are a particular case. \medskip

The following notion of quantum Latin squares
was introduced by Musto and Vicary \cite{MV15}. 
\begin{defi}
\label{defiLS}
A quantum Latin square of size $d$ 
is a square arrangement,
\begin{equation}
QLS(d)=
\begin{array}{ccc}
|\psi_{0,0}\rangle&\dots&|\psi_{0,d-1}\rangle\\
\vdots&&\vdots\\
|\psi_{d-1,0}\rangle&\dots&|\psi_{d-1,d-1}\rangle
\end{array}
\end{equation}
composed of $d^2$ single particle quantum states $|\psi_{ij}\rangle\in\mathcal{H}_d$, $i,j\in\{0,\dots,d-1\}$, such that each row and each column determines an orthonormal basis for a qudit system. 
\end{defi}
For instance, the following example of
a quantum Latin square was given in Ref. \cite{MV15},
\begin{equation}\label{LS-example}
\begin{array}{cccc}
|0\rangle&|1\rangle&|2\rangle&|3\rangle\\
|3\rangle&|2\rangle&|1\rangle&|0\rangle\\
|\chi_{-}\rangle&|\xi_{-}\rangle&|\xi_{+}\rangle&|\chi_{+}\rangle\\
|\chi_{+}\rangle&|\xi_{+}\rangle&|\xi_{-}\rangle&|\chi_{-}\rangle
\end{array},
\end{equation}
where two lower rows contain entangled states,
$|\chi_{\pm}\rangle=\frac{1}{\sqrt{2}}(|1\rangle\pm|0\rangle)$, $|\xi_+\rangle=\frac{1}{\sqrt{5}}(i\ket{0}+2\ket{3})$ and $|\xi_-\rangle=\frac{1}{\sqrt{5}}(2\ket{0}+i\ket{3})$.
As a first observation, we realize that any QLS is naturally related to a tripartite pure state having maximally mixed single particle reductions.
\begin{prop}\label{prop1}
A set of $d^2$ vectors $|\psi_{ij}\rangle\in\mathcal{H}_d$ forms a $\QLS(d)$ 
if and only if every single particle reduction of the three qudit state
\begin{equation}\label{PhiMOLS}
|\Phi\rangle=\sum_{i,j=0}^{d-1}|i\rangle|j\rangle|\psi_{ij}\rangle,
\end{equation}
is maximally mixed.
\end{prop}
\begin{proof}
Let $|\psi_{ij}\rangle\in\mathcal{H}_d$ be the $d^2$ entries of a $\QLS(d)$ and let us define the state $|\Phi\rangle=\sum_{i,j=0}^{d-1}|i\rangle|j\rangle|\psi_{ij}\rangle$. Therefore 
\begin{eqnarray}\label{rho_A}
\hspace{-0.2cm}\rho_A\!&=&\mathrm{Tr}_{BC}|\Phi\rangle\langle\Phi|\nonumber\\
&=&\mathrm{Tr}_{BC}\!\left(\sum_{i,j,i',j'=0}^{d-1}\!|ij\rangle_{AB}\langle i'j'|\otimes|\psi_{ij}\rangle_{C}\langle\psi_{i'j'}|\right)\nonumber\\
&=&\hspace{-0.2cm}\sum_{i,j,i'=0}^{d-1}\!\langle\psi_{ij}|\psi_{i'j}\rangle_{BC}|i\rangle_{A}\langle i'|=\hspace{-0.2cm}\sum_{i,j,i'=0}^{d-1}\!|i\rangle_{A}\langle i|=\mathbb{I}_d,\nonumber
\end{eqnarray}
where we used the fact that $|\psi_{ij}\rangle\in\mathcal{H}_d$ defines a $\QLS(d)$ and denoted $A,B,C$ for first, second and third party, respectively. Analogously, $\rho_B=\mathbb{I}_d$,
as we work with non--normalized states.
 Furthermore, we have
\begin{eqnarray}
\rho_C&=&\mathrm{Tr}_{AB}\!\left(\sum_{i,j,i',j'=0}^{d-1}\!|ij\rangle\langle i'j'|\otimes|\psi_{ij}\rangle\langle\psi_{i'j'}|\right)\nonumber\\
&=&\sum_{i,j=0}^{d-1}|\psi_{ij}\rangle\langle\psi_{ij}|=\mathbb{I}_d,\nonumber
\end{eqnarray}
and, therefore,  state (\ref{PhiMOLS}) has every single particle reduction maximally mixed. The reciprocal implication 
 works in the reverse way.
\end{proof}
Let us exemplify Prop. \ref{prop1} by considering the 1-uniform state of
a three qudit system,
\begin{equation}\label{phiGHZ}
|\phi\rangle=F_d\otimes F_d\otimes\mathbb{I}_d|GHZ_d\rangle=\sum_{l,m=0}^{d-1}|lm\rangle|\psi_{l,m}\rangle .
\end{equation}
Here  $|GHZ_d\rangle=\sum_{n=0}^{d-1}|nnn\rangle$
denotes a generalized $GHZ$ state of three subsystems with $d$ levels each,
 $F_d=\sum_{l,m=0}^{d-1}\omega^{lm}|l\rangle\langle m|$ 
is the discrete Fourier transform of size $d$
containing an unimodular number $\omega=e^{2\pi i/d}$ and the state reads
\begin{equation}\label{QLS}
|\psi_{l,m}\rangle=\sum_{n=0}^{d-1}\omega^{n(l+m)}|n\rangle .
\end{equation}
This construction works for any $d\geq2$. 
The $d^2$ states from Eq.(\ref{QLS}) determine a $QLS$ of size $d$, which is equivalent to the classical $[LS(d)]_{lm}=l+m|\!\mod d$ with $l,m=0,\dots d-1$, as 
the classical arrangement can be obtained 
by applying the same suitable local unitary operation to every column of the QLS. 
The state (\ref{phiGHZ}) is 1-uniform and it
is equivalent to the three-qudit GHZ state, in agreement with Proposition \ref{prop1}. 
Let us generalize this fact in the following observation.
\begin{obs}
A QLS(d) is equivalent to a classical LS(d) if and only if one arrangement can be transformed into the other by applying the same local unitary operation to every column.
\end{obs}
Furthermore, note that a unitary operation $U$ applied to a single column of a LS implies a controlled $U$ operation acting on the third party of the corresponding three-partite 1-uniform state (see Prop. \ref{prop1}). As consequence, the entanglement of the state is changed
and the Latin arrangement is spoiled by a single column unitary operation. \medskip

The notion of \emph{weekly orthogonal} QLS has been recently introduced \cite{M17}.
\begin{defi}\label{Mustodef}
A pair of QLS of size $d$ having entries $\{\varphi_{ij}\}$ and $\{\varphi^{\prime}_{ij}\}$ are weakly orthogonal when for every $i,j\in\{0,\dots,d-1\}$, there exists a unique $t_{ij}\in\{0,\dots,d-1\}$ such that
\begin{equation}\label{Mustoeq}
\sum_{l=0}^{d-1}\langle\varphi_{li}|\varphi^{\prime}_{lj}\rangle|l\rangle=|t_{ij}\rangle.
\end{equation}
\end{defi}
This definition reflects some desired aspects in orthogonal QLS. Indeed, it is reduced to standard definition of LS if the states $\varphi_{ij}$ belong to the computational basis. However, other fundamental ingredients seem to be missing here. For instance, the astonishing property that a pair of orthogonal QLS is \emph{not} necessarily equivalent to two QLS satisfying an orthogonality criteria, as we will see below. Those sets of orthogonal QLS that cannot be separated will be called \emph{essentially quantum Latin squares}. This new concept of non-separability of combinatorial designs is analogous to the non-separability of quantum states.

Let us now introduce the notion of orthogonality for QLS, which is \emph{not} 
equivalent to orthogonality for two separated quantum arrangements.
\begin{defi}\label{OLSdef}
A set of $d^2$ pure quantum states $|\psi_{i,j}\rangle\in \mathcal{H}_d^{\otimes 2}$ arranged as
\begin{equation}\label{MOQLS}
\begin{array}{ccc}
|\psi_{0,0}\rangle&\dots&|\psi_{0,d-1}\rangle\\
\vdots&&\vdots\\
|\psi_{d-1,0}\rangle&\dots&|\psi_{d-1,d-1}\rangle
\end{array}
\end{equation}
forms a pair of orthogonal quantum Latin squares (OQLS) if the following properties hold:
\begin{enumerate} 
\item The set of $d^2$ states $\{|\psi_{i,j}\rangle\}$ 
are orthogonal
and form a basis in $\mathcal{H}_d \otimes \mathcal{H}_d$.
\item The sum of every row in the array (\ref{MOQLS}), i.e.~$\sum_{j=0}^{d-1}|\psi_{i,j}\rangle$, is a 1-uniform state. 
\item The sum of every column in the array (\ref{MOQLS}), i.e.~$\sum_{i=0}^{d-1}|\psi_{i,j}\rangle$, is a 1-uniform state.
\end{enumerate}
\end{defi} 

\begin{obs}
Two OQLS composed of separable states, $|\psi^{AB}_{i,j}\rangle\!=\!|\eta^A_{ij}\rangle\otimes|\eta^B_{ij}\rangle$ for every $i,j\in\{0,\dots,d-1\}$, imply than both arrangements $\{|\eta^A_{ij}\rangle\}$ and $\{|\eta^B_{ij}\rangle\}$ determine QLS, according to Definition \ref{defiLS}. 
\end{obs}
Indeed, single party reductions to $A$ and $B$ of the states defined in
 items 2 and 3 above are proportional to the maximally mixed state, so that every 
row and every column of arrangements $\{|\eta^A_{ij}\rangle\}$ and $\{|\eta^B_{ij}\rangle\}$ form an orthonormal basis. Moreover, if entries of each QLS are given by elements of the computational basis then Definitions \ref{defiLS} and \ref{OLSdef} reduces to the classical definition of LS and OLS, respectively (see Section \ref{S2}).\medskip

As we will show in Section \ref{S4}, OQLS are closely related to 2-uniform states. In order to achieve higher classes of multipartite entanglement, i.e. $k$-uniformity for $k>2$, 
one has to generalize  quantum combinatorial arrangements to higher dimensions.
To this end, let us go a step forward and introduce quantum Latin cubes.
\begin{defi}
A quantum Latin cube (QLC) of size $d$ is a
 cubic arrangement composed of $d^3$ single particle quantum pure 
states $|\psi_{ijk}\rangle\in\mathcal{H}_d$, $i,j,k\in\{0,\dots,d-1\}$,
 such that every row, every column and every file form a set of orthogonal states. 
\end{defi}
For instance, in the case of a cubic arrangement composed by qubit quantum states, i.e. $d=2$, we have the cube (\ref{OA8422}).
Let us introduce a notion of orthogonality between cubic arrangements. 

\begin{defi}\label{OLCdef}
A set of $d^3$ tri-partite pure states $|\psi_{x,y,z}\rangle$ belonging 
to a composed Hilbert space $\mathcal{H}_d^3$, arranged as
\tikzset{
  vertex/.style={
    circle, minimum size=20pt, inner sep=0pt, fill=white},
  axial/.style={
    rectangle, minimum size=20pt, inner sep=0pt, fill=gray!30},
  edge/.style={draw,thick,-,black},
  rotu/.style={blue, midway, densely dotted},
  sinal/.style={draw, circle, inner sep=0pt, thin},
}

\def\dist{0.4}
\begin{tikzpicture}[
  scale=4, -, z={(0.55,0.45)}, node distance=0.65cm]
  \node[vertex] (v0) at (0,0,0) {$\ket{\psi_{0,d-1,0}}$};
  \node[vertex] (v1) at (0,1,0) {$\ket{\psi_{0,0,0}}$};
  \node[vertex] (v2) at (1,0,0) {$\ket{\psi_{0,d-1,d-1}}$};
  \node[vertex] (v3) at (1,1,0) {$\ket{\psi_{0,0,d-1}}$};
  \node[vertex] (v4) at (0,0,1) {$\ket{\psi_{d-1,d-1,0}}$};
  \node[vertex] (v5) at (0,1,1) {$\ket{\psi_{d-1,0,0}}$};
  \node[vertex] (v6) at (1,0,1) {$\ket{\psi_{d-1,d-1,d-1}}$};
  \node[vertex] (v7) at (1,1,1) {$\ket{\psi_{d-1,0,d-1}}$};
  \draw[edge] (v0) -- (v1) node[pos=0.5, anchor=center, fill=white] {$\vdots$} 
  -- (v3) node[pos=0.3,fill=white,rectangle, anchor=center] {$\ldots$} 
  -- (v2) node[pos=0.4,fill=white, rectangle, anchor=center] {$\vdots$}
  -- (v0) node[pos=0.5,fill=white] {$\ldots$};
  \draw[edge] (v0) -- (v4) node[pos=0.5,fill=white] {$\udots$}
  -- (v5) node[pos=0.35,fill=white] {$\vdots$}
  -- (v1) node[pos=0.5,fill=white] {$\udots$};
  \draw[edge] (v2) -- (v6) node[pos=0.5,fill=white] {$\udots$}
  -- (v7) node[pos=0.5,fill=white] {$\vdots$}
  -- (v3) node[pos=0.5,fill=white] {$\udots$};
  \draw[edge] (v4) -- (v6) node[pos=0.65,fill=white] {$\ldots$};
    \draw[edge] (v5) -- (v7) node[pos=0.5,fill=white] {$\ldots$};
\end{tikzpicture}\label{OQLC}
forms a triple of mutually orthogonal quantum Latin cubes (MOQLC)
 if the following properties hold:
\begin{enumerate} 
\item  The set of $d^3$ states $\{|\psi_{x,y,z}\rangle\}$ are orthogonal.
\item The sum of every row in the array (\ref{OQLC}), i.e.~$\sum_{i=0}^{d-1}|\psi_{x,y,z}\rangle$, is a 1-uniform state. 
\item The sum of every column in the array (\ref{OQLC}), i.e.~$\sum_{j=0}^{d-1}|\psi_{x,y,z}\rangle$, is a 1-uniform state. 
\item The sum of every file in the array (\ref{OQLC}), i.e.~$\sum_{k=0}^{d-1}|\psi_{x,y,z}\rangle$, is a 1-uniform state. 
\end{enumerate}
\end{defi}
Analogously to Definition \ref{OLSdef}, if the $d^3$ states forming
 a set of MOQLC are fully separable, i.e.
 $|\psi^{ABC}_{x,y,z}\rangle=|\eta^A_{x,y,z}\rangle\otimes|
\eta^{B}_{x,y,z}\rangle\otimes|\eta^{C}_{x,y,z}\rangle$,
 then each set of states $\{\eta^A_{x,y,z}\rangle\}$, 
$\{\eta^{B}_{x,y,z}\rangle\}$ and $\{|\eta^{C}_{x,y,z}\rangle\}$
 forms a QLS according to Definition \ref{defiLS}. Furthermore, in such a case
 a fully separable MOQLC is equivalent to a classical MOLC, in the sense that one can be connected to the other by applying local 
unitary operations acting in columns of the arrangements. This is so because any 
single-party orthonormal basis can be transformed into the computational basis 
by applying a suitable local unitary transformation. Also, if the states forming 
the cube (\ref{OQLC}) are biseparable with respect to a given partition, 
e.g. $|\psi^{ABC}_{x,y,z}\rangle\!=\!|\eta^A_{x,y,z}\rangle\otimes|\eta^{BC}_{x,y,z}\rangle$
 for every $x,y,z\in\{0,\dots,d-1\}$, then the single-party arrangement 
$\{|\eta^A_{x,y,z}\rangle\}$ defines a QLC according to Definition \ref{OLCdef}. 
It is important here to note that the bipartite arrangement 
$\{|\eta^{BC}_{x,y,z}\rangle\}$ \emph{not necessarily} forms a pair of OQLC.
 This surprising fact is closely related to the lack of some classes of 
multipartite absolutely maximal entanglement, e.g. AME($N$,2) 
states exist only if the number of qubits is given by $N=2,3,5,6$  \cite{S04,HGS17,HESG17}.
\medskip

As the concepts of OQLS and OQLC are settled, let us define an arbitrary 
dimensional kind of quantum combinatorial arrangements, 
called \emph{quantum Latin hypercubes}. These quantum arrangements can be
 connected to $k$-uniform states for $N$ qudit systems having $d$ levels 
each for any $k$, $N$ and $d$, as we will show in Section \ref{S4}.
\begin{defi}
\label{defiLH}
A quantum Latin hypercube of size $d$ and dimension $k$, denoted $QLH(d,k)$, 
is an arrangement composed of $d^k$ single particle quantum states 
$|\psi_{i_1,\dots,i_k}\rangle\in\mathcal{H}^{\otimes k}_d$, 
$i_1,\dots, i_k\in\{0,\dots,d-1\}$, 
 such that all states belonging to an edge of the hypercube are orthogonal. 
\end{defi}
In particular, 
for $k=2$ quantum hypercube $QLH(d,2)$ reduces to the square $QLS(d)$, 
while for $k=3$ they form a cube, $QLH(d,3)=QLC(d)$.
For instance, the state AME(8,7) with minimal support
determines an $m=4$ hypercubes MOQLH of size $d=7$ in dimension $k=4$,
equivalent to 4 classical MOLH. 
 These four separable hypercubes can be easily found from the 
 orthogonal array OA($7^4,8,7,4$)  with index equal to unity,
 associated to the AME(8,7) state -- see Theorem 3 and Prop. 2 (i) in \cite{GZ14}.
 This state is also related to a maximum distance separable (MDS) code
 \cite{Helwig-Cui2013,Raissi-2017}. 
Furthermore, the AME(8,5) state, which has non-minimal support, defines  $m=4$ -- essentially quantum -- orthogonal Latin hypercubes in dimension $k=4$, 
 with entangled entries. 
This state can be constructed from \emph{ququint codes} \cite{GR15}.\bigskip 

We can extend the sets of OQLS and OQLC to sets of $m$ mutually orthogonal quantum 
Latin hypercubes (MOQLH) of size $d$ and dimension  $k\le m$. The following definition contains all previously defined combinatorial designs.
\begin{defi}\label{defi6}
A set of $m$ mutually orthogonal quantum Latin hypercubes of size 
$d$ in dimension $k$, denoted $m$ $MOQLH(d)$, is a $k$-dimensional arrangement
 composed of $m$-qudit states 
$|\psi_{i_1,\dots,i_k}\rangle\in\mathcal{H}^{\otimes m}_d$,
 $i_1,\dots,i_m\in\{0,\dots,d-1\}$ such that the following properties hold:
\begin{enumerate} 
\item The set of $d^k$ states $\{|\psi_{i_1,\dots,i_k}\rangle\}$ are orthogonal.
\item The sum of $d$ states belonging to the same edge of the hypercube, 
i.e.~$\sum_{i_s=0}^{d-1}|\psi_{i_1,\dots,i_s,\dots,i_m}\rangle$ for every $1\leq s\leq m$,
forms a 1-uniform state. 
\end{enumerate}
\end{defi}
In particular, a set of $m$ MOLS are also MOQLS, e.g. the classical arrangements (\ref{MOLS3}) agree Definition \ref{defi6}. In Section \ref{S4}, we introduce a suitable tool to generate quantum Latin arrangements, called \emph{quantum orthogonal arrays}, and also establish its connection with quantum Latin arrangements.
\medskip

\subsection{Bounds for MOQLH}
Let us now study upper bounds for the maximal number of classical and quantum Latin arrangements. The theory of orthogonal arrays provides a bound \cite{B52}
for the maximal number of columns of an OA($d^k,2+m_C,d,k$), that has index unity. 
Therefore, it is easy to derive an upper bound for the maximal
 allowed number $m_C$ of classical MOLH of size $d$ and dimension $k$:
\begin{equation}\label{Bush}
m_C\leq \left\{
\begin{array}{c l}
k-1 & \mbox{if } d\leq k\\
d+k-4 & \mbox{if } d> k\geq3\\
d+k-3 & \mbox{in all other cases.}
\end{array}
\right.
\end{equation}
For example, in dimension $k=2$ we have that $m$ MOLS of size $d$ can only exist for $m_C\leq d-1$, for any $d\geq2$. The upper bound $m=d-1$ can be saturated for $d$ being a prime power number. These results, well-known in standard combinatorics, 
motivate us to derive similar results for quantum Latin arrangements. However,
 derivation of such a generalized bound requires solving a complicated optimization problem
 formalized by Scott -- see Eqs.(39)--(41) in Ref. \cite{S04}. 
Given the set of parameters $N,d,k$ ($n,D,d$ in the original notation) these equations can be solved by considering linear programming techniques. The particular case $k=\lfloor N/2\rfloor$, 
for which the arrangements are associated to AME states, 
can be analytically solved. Therefore, we are able to provide an analytic bound for the maximal number $m_Q$ of MOQLH in the case of maximal possible dimension $k=\lfloor N/2\rfloor$ as follows: 
\begin{equation}\label{Scott}
m_Q\leq \left\{
\begin{array}{c l}
2(d^2-1) & \mbox{ if } N \mbox{ is even}\\
2d(d+1)-1 & \mbox{ if } N \mbox{ is odd.}
\end{array}
\right.
\end{equation}
For instance, for $N=4$ and $k=2$ we have that $m_Q~\leq~2(d^2-1)$ MOQLS exist for any size $d$, which is $2(d+1)$ times larger than the classical bound $m_C\leq d-1$. It is important to note that bounds  (\ref{Scott}) are not tight, as the bounds provided by Scott
 \cite{S04} are not tight -- see also \cite{HESG17}.\medskip

Inequalities (\ref{Bush}) and (\ref{Scott}) can be useful to detect genuine quantumness in MOQLH. In general, given a set of $m$ MOQLH it is hard to detect inequivalence to a classical set of MOLS. Typically, such kind of comparison would require to consider a full set of entanglement invariants. However, for those cases where $m>m_C$ it is ensured that a MOQLH is essentially quantum. For instance, a single LS of size two exists and there are no two QOLS of size two. Surprisingly, there exists \emph{three} entangled MOQLS of size two exist, as we will show in Section \ref{S4}.

\section{Quantum orthogonal arrays}\label{S4}
In this section, we introduce quantum orthogonal arrays. This concept allows us to derive a simple rule to generate infinitely many classes of $k$-uniform states and absolutely maximally entangled states, in particular.
\begin{defi}\label{QOAdef}
A quantum orthogonal array $\QOA(r,N,d,k)$ is an arrangement consisting of $r$ rows composed by $N$-partite normalized pure quantum states $|\varphi_j\rangle\in\mathcal{H}_d^{\otimes N}$, having $d$ internal levels each, such that 
\begin{equation}
k\sum_{j=0}^{r-1}\mathrm{Tr}_{i_1,\dots,i_{N-k}}\bigl(|\varphi_{j}\rangle\langle\varphi_{j}|\bigr)=r\,\mathbb{I}_{k},
\end{equation}
for every subset of $N-k$ parties $\{i_1,\dots,i_{N-k}\}$.
\end{defi}
In words, a QOA is an arrangement having $N$ columns, possibly entangled, such that every reduction to $k$ columns defines a \emph{Positive Operator Valued Measure} (POVM). We recall that a POVM is a set of positive semidefinite operators such that they sum up to identity, determining a generalized quantum measurement \cite{NC00}.

We can also provide a connection to error correction codes that suggest us to consider generalized measurements instead of projective measurements in QOA. 
Note that any AME state (or $k$-uniform state) 
can be related to a certain quantum error correction code \cite{S04}.
In particular, an AME state of $N$ parties with local dimension $d$, 
corresponds to a quantum code -- which can be considered as an injective mapping from 
the space of $K=1$ messages to a subset $C$ of the set of
codewords with length $N$ -- often denoted by $(\!(N,K=1,D=\floor{N/2}+1)\!)_d$.
In this notation, the parameter $D$ is the distance of the code, i.e. the minimal number of
local operations performed on a single qudit that 
are needed to create a non-zero overlap between two codewords \cite{Gottesman-thesis}.
Knill-Laflamme theorem \cite{KL97} implies that
a subspace $C$ of the Hilbert space $\hiH = \mathbb{C}_d^{\otimes N}$ generates an error correcting quantum code,
 if there exist recovery operators $R_1, R_2, \ldots$ such that for 
any state $\rho$ with support in $C$ and any collection of error operators $A_1,A_2,\ldots$  with $\sum_e E_e^{\dagger}E_e=\1$, we have $\sum_{r,e}R_r E_e \rho E_e^{\dagger}R_r ^{\dagger}=\rho \otimes \1$.
In this case $R_1,R_2, \ldots$ are a finite sequence of operators in $\hiH$ satisfying the relation $\sum_r R_r^{\dagger}R_r=\1$.
This theorem combined with the fact that an AME state yields an error correction code 
allows us to define quantum orthogonal arrays 
in a way that every reduction produces a POVM.

\medskip

Definition \ref{QOAdef} forms a natural extension of the classical concept of orthogonal arrays
to quantum theory: { the classical digits from $(0,\dots,d-1)$ are 
generalized to quantum states from ${\cal H}_d$, while
the classical concept of subsets of columns are replaced by partial trace}.

From now on, we assume that columns of quantum arrangements are connected by the Kronecker product. Also, QOA having the minimal possible number of rows, i.e.  $r=d^k$, are called \emph{index unity}, as occurs in the classical case.

Let us introduce equivalent classes of QOA as a natural generalization of its classical counterpart, defined in Section \ref{S2}. Two QOA are \emph{equivalent} if one can 
transform one arrangement into the other one by applying suitable local unitary operations
to columns and permutation of rows or columns. 
Note that permutation of columns in quantum states produce states inequivalent under LOCC, in general. Nevertheless, as interchange of particles does not change the amount of entanglement in quantum states, from now on we will restrict our attention to QOA inequivalent under swap operations. Note that the only allowed local unitary operations in classical OA are permutation matrices, equivalent to relabelling of symbols. In contrast to quantum Latin arrangements, 
in QOA we are allowed to apply any local unitary operation to any column without spoiling the orthogonal array. 
To illustrate these ideas let us consider the following example:
\begin{equation*}\label{LS}
(\mathbb{I}\otimes\sigma_x)\,\begin{array}{cc}
|0\rangle&|0\rangle\\
|1\rangle&|1\rangle\\
\end{array}=\begin{array}{cc}
|0\rangle&|1\rangle\\
|1\rangle&|0\rangle\\
\end{array},
\end{equation*}
where $\sigma_x=\{\{0,1\},\{1,0\}\}$ is the Pauli shift operator. In this way, we obtain two equivalent classical OA. Instead, by applying the Hadamard gate $H=\{\{1,1\},\{1,-1\}\}$ to the second column, i.e.,
\begin{equation}\label{3MOQLS2}
(\mathbb{I}\otimes H)\,\begin{array}{cc}
|0\rangle&|0\rangle\\
|1\rangle&|1\rangle\\
\end{array}=\begin{array}{cc}
|0\rangle&|+\rangle\\
|1\rangle&|-\rangle\\
\end{array},
\end{equation}
with $|\pm \rangle  =|0 \rangle \pm |1 \rangle $,
we obtain a QOA which is equivalent under local unitary operations to a classical OA. The simplest essentially quantum orthogonal array consists of five columns, 
\begin{equation}\label{QOA4522}
QOA(4,5,2,2) =
\begin{array}{cccc}
|0\rangle&|0\rangle&|0\rangle&|\Phi^+\rangle\\
|0\rangle&|1\rangle&|1\rangle&|\Psi^+\rangle\\
|1\rangle&|0\rangle&|1\rangle&|\Psi^-\rangle\\
|1\rangle&|1\rangle&|0\rangle&|\Phi^-\rangle\\
\end{array},
\end{equation}
where, $|\Phi^{\pm}\rangle=(|00\rangle\pm|11\rangle)/\sqrt{2}$ and $|\Psi^{\pm}\rangle=(|01\rangle\pm|10\rangle)/\sqrt{2}$ denote the Bell basis.
To emphasize that some of these columns are separable (classical) 
and some of them are entangled (quantum),
we shall also write  $QOA(4,3_C+2_Q,2,2)$, as the second argument
denotes three classical and two quantum columns. Note that the number of classical and quantum columns, i.e. $N_C$ and $N_Q$ such that $N=N_C+N_Q$, are invariant under local unitary operations acting on columns of the QOA. Moreover, a QOA is equivalent to a classical OA if and only if $N_Q=0$, thus also implying a classical set of MOLS and a classical error correction code \cite{HSS99}. Roughly speaking, parameter $N_Q$ quantifies how much quantum is a given QOA and its related MOQLS and error correction code. 
As a further comment, note that every reduction to two columns of the arrangement (\ref{QOA4522}) form a POVM, 
where partial trace should be considered for entangled columns. The fact that QOA (\ref{QOA4522}) is not equivalent to a classical OA is in correspondence with the fact that AME(5,2) state cannot be written as a convex combination of elements of the 5-qubit computational basis. 

As we have seen in Section \ref{S2}, OLS arise from OA. First two columns of the OA provide address to entries of the first and second LS, whose values are determined by the third and fourth column of the OA, see Eq. (\ref{MOLS3}). In the same way, from $\QOA(4,5,2,2)$ of Eq.(\ref{QOA4522}) we derive three MOQLS of size 2, which are essentially quantum. A triple of mutually orthogonal 
quantum Latin squares reads, 
\begin{equation}\label{3MOQLS2}
MOQLS(2) =
\begin{array}{cc}
|0\rangle|\Phi^+\rangle&|1\rangle|\Psi^+\rangle\\
|1\rangle|\Psi^-\rangle&|0\rangle|\Phi^-\rangle
\end{array}.
\end{equation}
First two columns of QOA (\ref{QOA4522}) address entries of the three MOQLS (\ref{3MOQLS2}). Note that these three MOQLS are entangled, which is a direct consequence of the fact that QOA (\ref{QOA4522}) is not equivalent to a classical one. Indeed, QOA (\ref{QOA4522}) contains entangled columns. According to the results shown in Section \ref{S3}, a single party arrangement belonging to a set of MOQLS determines a QLS, what can be seen from Eq.(\ref{3MOQLS2}) after tracing out second and third party. However, the bipartite arrangement obtained from taking partial trace over the first subsystem of the QOA (\ref{3MOQLS2}), i.e.  
\begin{equation}
\begin{array}{cc}
|\Phi^+\rangle&|\Psi^+\rangle\\
|\Psi^-\rangle&|\Phi^-\rangle
\end{array},
\label{Q2x2}
\end{equation}
is \emph{not} a pair of orthogonal QLS. This is simple to observe if we take into account Definition \ref{OLSdef}. Indeed, the sum of every column of the arrangement (\ref{Q2x2}) determines a 1-uniform state but sum of every row gives
 a separable state. It is possible to prove that such $\QOA(r,4,2,2)$ 
does not exist for any $r\in\mathbb{N}$, which is related to the fact that an $\AME(4,2)$ state does not exist \cite{HS00}.

As a further example, we consider 
the following array consisting of three classical and three quantum columns, 
\begin{equation}\label{QOA8623}
\QOA(8,3_C+3_Q,2,3) =
\begin{array}{cccc}
|0\rangle&|0\rangle&|0\rangle&|GHZ_{000}\rangle\\
|0\rangle&|0\rangle&|1\rangle&|GHZ_{001}\rangle\\
|0\rangle&|1\rangle&|0\rangle&|GHZ_{010}\rangle\\
|0\rangle&|1\rangle&|1\rangle&|GHZ_{011}\rangle\\
|1\rangle&|0\rangle&|0\rangle&|GHZ_{100}\rangle\\
|1\rangle&|0\rangle&|1\rangle&|GHZ_{101}\rangle\\
|1\rangle&|1\rangle&|0\rangle&|GHZ_{110}\rangle\\
|1\rangle&|1\rangle&|1\rangle&|GHZ_{111}\rangle
\end{array},
\end{equation}
that produces three MOQLC of size 2:\vspace{0.3cm}
\begin{equation}\label{cube}
\begin{array}{l}
MOQLC(2)=\\ [0.3cm]
\begin{tikzpicture}
[
back line/.style={dashed},
cross line/.style={preaction={draw=white, -,
line width=6pt}}]
\matrix (m) [matrix of math nodes,
row sep=2.9em, column sep=0.1em,
text height=0.25ex,
text depth=0.25ex]
{
& \ket{GHZ_{100}} & & \ket{GHZ_{101}} \\
\ket{GHZ_{000}} & & \ket{GHZ_{001}} \\
& \ket{GHZ_{110}} & & \ket{GHZ_{111}} \\
\ket{GHZ_{010}} & & \ket{GHZ_{011}} \\
};
\path[-]
(m-1-2) edge [back line] (m-1-4)
edge [back line] (m-2-1)
edge [back line] (m-3-2)
(m-1-4) edge [back line] (m-3-4)
edge [back line] (m-2-3)
(m-2-1) edge [back line] (m-2-3)
edge [back line] (m-4-1)
(m-3-2) edge [back line] (m-3-4)
edge [back line] (m-4-1)
(m-4-1) edge [back line] (m-4-3)
(m-3-4) edge [back line] (m-4-3)
(m-2-3) edge [back line] (m-4-3);
\end{tikzpicture}
\end{array}
\end{equation}
Here, the tri-partite orthonormal basis is composed by
eight states locally equivalent to the 3-qubit GHZ state,
$|GHZ\rangle=|000\rangle + |111\rangle$.
These states form an orthonormal basis in 
${\cal H}_8={\cal H}_2\otimes {\cal H}_2\otimes {\cal H}_2$,
\begin{equation}\label{GHZ}
|GHZ_{ijk}\rangle =  (-1)^{\alpha_{ijk}}\sigma_{i}\otimes\sigma_{j}\otimes\sigma_{k} |GHZ\rangle,
\end{equation}
where $i,j,k=\{0,1\}$ and $\sigma_0$ and $\sigma_1$ represent the Pauli matrices $\sigma_x$ and $\sigma_z$, respectively. Global phases given by $\alpha_{ijk}=1$ if $i=j=k$ and $\alpha_{ijk}=0$ otherwise are added to states (\ref{GHZ}) forming the GHZ basis, in such a way that the construction (\ref{cube}) forms a quantum Latin cube.\medskip

Let us show that a $\QOA(r,N,d,k)$ determines a
 $k$-uniform state of $N$ qudits, in the same way as  
an irredundant $\OA(r,N,d,k)$ implies a $k$-uniform
 state of $N$ subsystems with $d$ levels each \cite{GZ14}.

\begin{prop}
\label{prop2}
The sum of rows of a $\QOA(r,N,d,k)$ produces a $k$-uniform state of 
a quantum system composed of $N$ parties with $d$ levels each.
\end{prop}
\begin{proof}
Every reduction to $k$ columns of a $\QOA(r,N,d,k)$ defines a POVM, and thus the sum of its elements produces the identity operator.   
\end{proof}
For instance, $\QOA(4,5,2,2)$ of Eq.(\ref{QOA4522}), related to the squares
 (\ref{3MOQLS2}), produces the $2$--uniform five-qubit state \cite{LMPZ96}
\begin{eqnarray}\label{AME52}
\AME(5,2) &=&|000\rangle|\Phi^{+}\rangle+|011\rangle|\Psi^{+}\rangle+\nonumber\\
&&|101\rangle|\Psi^{-}\rangle+|110\rangle|\Phi^{-}\rangle .
\end{eqnarray}
Furthermore, the array $\QOA(8,6,2,3)$ presented in Eq.(\ref{QOA8623}), 
and related to the cube (\ref{cube}), produces the AME state for six-qubit systems 
\cite{BPBZCP07},
\begin{equation}\label{AME62}
\AME(6,2)=\sum_{x,y,z=0}^1|x,y,z\rangle|GHZ_{xyz}\rangle.
\end{equation}
Proposition \ref{prop2} reveals that QOA generalizes the notion of irredundant OA and \emph{not} the entire set of OA. For instance, the non-irredundant classical array, 
 \begin{equation}\label{OA4321}
\OA(4,3,2,1) =\begin{array}{ccc}
0&0&0\\
0&1&0\\
1&0&1\\
1&1&1
\end{array},
\end{equation}
is not equivalent to a $\QOA(r,3,2,1)$ for any $r$. This is so because OA (\ref{OA4321}) does not produce a 1-uniform state and, by definition, any QOA produces at least a 1-uniform state. The key difference existing between classical and quantum OA relies on the fact that the action of removing columns in classical OA is \emph{not} equivalent to taking
the partial trace in the quantum case.
Precisely, these operations are equivalent only if the orthogonal array considered
 is irredundant. Furthermore, the juxtaposition 
of two OA is still an OA, whereas the same statement does not hold for QOA. 
This is connected to the fact that the sum of two $k$-uniform states is not necessarily a $k$-uniform (see in Section \ref{S6}).
 Nonetheless, all classical $\OA(d^k,2+m,d,2)$, 
associated to $m$ mutually orthogonal hypercubes size $d$ are irredundant \cite{GZ14}. Thus, any set of $m$ mutually orthogonal Latin hypercubes, in particular any set of $m$ MOQLS, is linked to a QOA -- see Fig. \ref{Fig2}. As a natural generalization of this result, we have the following proposition.
\begin{prop}\label{prop3}
A $\QOA(d^k,k+m,d,k)$ generates $m$ MOQLH of size $d$ in dimension $k$. 
\end{prop}
We generate MOQLH from QOA in the same way than MOLH arise from classical OA. That is, first $k$ classical columns of a QOA address the location of entries and the remaining $m$ 
columns determine the values of every entry of the quantum Latin arrangement.\medskip
 
Let us discuss some important open issues. The lowest dimensional open case for MOQLS occurs for $k=m=2$ and $d=6$, 
that is, two OQLS of size six. 
It is well-known that the classical problem of 
\emph{36 officers of Euler} has no solution \cite{E82},
as there are no orthogonal Latin squares of order six.
 After an exhaustive numerical exploration 
we are tempted to advance the following conjecture.
\begin{conje}\label{conje2}
Two orthogonal QLS of size $6$ do not exist.
\end{conje}
This conjecture is equivalent to say that the famous problem of Euler has no
solution also in the generalized quantum setup, as $36$ officers 
are now allowed to be described by entangled quantum states.

It also would imply a negative answer to the existence of AME state for a system composed of four systems with $6$ levels each -- compare 
related studies in Refs. \cite{GZ14,Hu17}. The existence of AME(4,6) state currently represents the lowest dimensional open case, and the only open case in the family of states AME(4,$d$). We recall that AME(4,$d$) exist for any $d\neq2,6$. Indeed, all of these states have minimal support and can be easily generated from two classical MOLS($d$), equivalently from OA($d^2$,4,$d$,2) \cite{GZ14}. \bigskip

Let us now relate quantum Latin arrangements defined through QOA with those established in Definition \ref{defi6}. The special subset of MOQLH satisfying Definition \ref{defi6} produce highly  entangled $k$-uniform states, (e.g. cluster states),
robust under the presence of a noisy environment.
Indeed, we might interpret the hyper-faces of MOQLS as a protection of multipartite entanglement contained in lower dimensional faces of the hypercube. For instance, the generalized,
$N$--qudit GHZ state, $\sum_{i=0}^{d-1}|i\rangle^{\otimes N}$,
 defines the following set of $N$ MOQLH of size $d$ defined in dimension $k=1$:
\begin{equation}\label{GHZNd}
|\underbrace{0,\dots,0}_{N}\rangle\, -- \,|\underbrace{1,\dots,1}_{N}\rangle\, -- \,|\underbrace{d-1,\dots,d-1}_{N}\rangle.
\end{equation}
Here, the double line ($--$) denotes edges in the same way as depicted before, c.f. Definition \ref{OLCdef}. Arrangement (\ref{GHZNd}) has a unique 1-dimensional face, evidencing
 fragility of entanglement of GHZ states with respect to the noisy environment. 
On the other hand, the square (\ref{MOQLS}) produces a state having a higher robustness, as the square transforms to an edge under the presence of a local measurement on any of its parties. For instance, the 5-qubit state produced by three MOQLS of size $d=2$, see Eq.(\ref{3MOQLS2}), defines a \emph{perfect code} for quantum error correction \cite{LMPZ96}.

In order to understand robustness of entanglement produced by states coming from Definition \ref{defi6}, we need to recall two quantifiers of robustness \cite{BR01}: \medskip

\noindent\emph{Maximum connectedness $(\mathcal{C})$:} a multipartite quantum state is maximally connected if any two qudits can be projected, with certainty, into a Bell state by implementing local measurements on the complementary subset of parties. \medskip

\noindent\emph{Persistency of entanglement $(\mathcal{P})$:} the minimal number of local measurements to be implemented such that, for all measurement outcomes, 
the state is completely disentangled.\medskip

Now we are in position to establish the following result.
\begin{prop}\label{prop4}
A set of $m$ MOQLH $\{|\varphi_{i_1,\dots,i_{m}}\rangle\}$ of size $d$ defined in dimension $k$, composed by $d^k$ states of $m$ qudit systems having $d$ levels each, defines a $k$-uniform state for $N=k+m$ qudit systems, given by
\begin{equation}\label{robust}
|\phi\rangle=\sum_{i_1,\dots,i_{k}=0}^{d-1}|i_1,\dots,i_{k}\rangle|\varphi_{i_1,\dots,i_{m}}\rangle.
\end{equation} 
Even more, if  $k^{\prime}\leq k$ subsystems belonging to the first $k$ qudits are measured then the remaining entangled state is $(k-k^{\prime})$-uniform. In particular, if a state $|\phi\rangle$ can be written in the form (\ref{robust}) for its $\binom{N}{k}$ possible bipartitions of $k$ parties out of $N$ then it has maximum connectedness $\mathcal{C}=k-1$ and persistency of entanglement $\mathcal{P}\geq k$. 
\end{prop}
\begin{proof}
The state $|\phi\rangle$ defined for $N=k+m$  subsystems with $d$ levels each 
is $k$-uniform, since the following two facts hold: 
\emph{(i)} the set of $m$ MOQLH defined in dimension $k$ define a QOA($d^k,N,d,k$) and \emph{(ii)} Prop. \ref{prop2} applies. The fact that maximum connectedness is at least $\mathcal{C}=k-1$ comes straight from Property \emph{2} in Definition \ref{defi6}. By the same reason, we have $\mathcal{P}\geq k$, as an additional measurement may possibly destroy the $1$-uniformity of the remaining $m$-partite entangled states $|\phi^{\prime}\rangle=\sum_{i_1,\dots,i_{k}=0}^{d-1}|\varphi_{i_1,\dots,i_{m}}\rangle$.
\end{proof}
For instance, the state AME(5,2) defined in (\ref{AME52}), constructed through MOQLS (\ref{3MOQLS2}), satisfies $\mathcal{C}=1$, and defines a $1$-dimensional 
subspace protected under decoherence \cite{LMPZ96}.\bigskip

Let us summarize some important connections existing between classical and quantum arrangements and $k$-uniform states derived along this section. First, we start considering previously known connections. The following standard ('classical') notions are equivalent:
\begin{enumerate}
\item [\emph{1.}] QOA with fully separable columns ($\equiv$ OA)

[e.g. QOA$(9,4_C+0_Q,3,2)\equiv$OA(9,4,3,2) in Eq.(\ref{MOLS3})]
\item [\emph{2.}] Sets of $m$ separable MOQLH($d$) in dimension $k$\\ ($\equiv$ MOLH)
[e.g. classical LS$_A$ and LS$_B$ in Eq.(\ref{MOLS3})]
\item [\emph{3.}] $N$ qudit $k$-uniform states with minimal support

[e.g. AME(4,3) state in Eq.(\ref{AME43})]
\end{enumerate}
Here, the symbol $\equiv$ denotes  equivalence under local unitary operations applied to columns 
of an array.
  Connection \emph{1}-\emph{2} is well known in mathematics since the early times of orthogonal arrays theory  -- see Chapter 8 in Ref. \cite{HSS99}. 
Connections \emph{1}-\emph{3} and \emph{2}-\emph{3} have been recently established, 
see Refs. \cite{GZ14} and \cite{GALRZ15}, respectively. 
Furthermore, in the case of $N=2k$ there exists a link between AME states and 
multi-unitary permutation matrices \cite{GALRZ15}. 

In a similar manner, the following generalized ('quantum') notions are equivalent, 
\begin{enumerate}
\item [\emph{a.}] QOA with entangled columns ($\not\equiv$ OA)

[e.g. Eqs.(\ref{QOA4522}) and (\ref{QOA8623})]
\item [\emph{b.}] Entangled MOQLH ($\not\equiv$ fully separable MOQLH)

[e.g. Eqs. (\ref{3MOQLS2})]

\item [\emph{c.}] $N$ qudit $k$-uniform states with non-minimal support.
($\not\equiv$ to minimal support states)

[e.g. Eqs.(\ref{AME52}) and (\ref{AME62})]
\end{enumerate}
The above  relations \emph{a}-\emph{b}, \emph{a}-\emph{c} and \emph{b}-\emph{c} 
form a novel contribution of the present work.
A further connection to general multi-unitary matrices 
 occurs when $N=2k$ \cite{GALRZ15}. 
 
Note that a QOA having at least one pair of entangled columns necessarily implies 
existence of entangled OQLS that cannot be separated, 
in the same way as entangled states cannot be represented as
the tensor product of two single party pure states.\medskip

\begin{figure}[!h]
\centering
{\includegraphics[width=8.5cm]{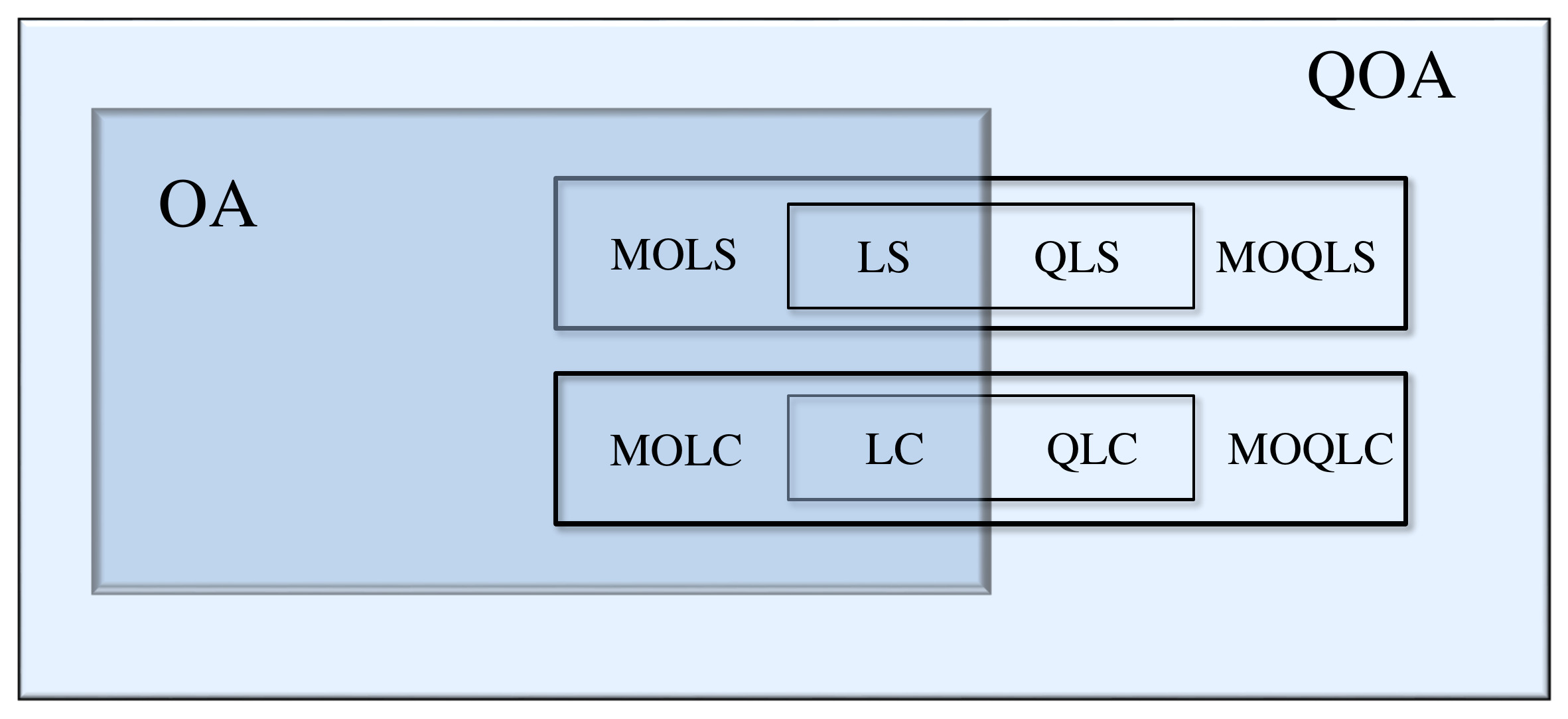}} 
\caption{Generalization of orthogonal arrays (OA) to quantum orthogonal arrays (QOA). This extension allows us to naturally generalize some classical arrangements to quantum mechanics: Quantum Latin squares (QLS), Quantum Latin cubes (QLC) and Mutually orthogonal quantum arrangements (MOQLS and MOQLC).}
\label{Fig2}
\end{figure}

\section{QOA and multiunitary matrices}\label{S5}
Let us consider a quantum system consisting of $N=2k$ parties having $d$ level systems each, where $k\ge1$ and the system is prepared in the pure state
\begin{equation}\label{phia}
|\phi\rangle=\hspace{-0.2cm}\sum_{n_1,\dots,n_k\atop\nu_1,\dots,\nu_k}a_{n_1,\dots,n_k\atop\nu_1,\dots,\nu_k}|n_1,\dots,n_k,\nu_1,\dots,\nu_k\rangle,
\end{equation} 
where every sum goes from $0$ to $d-1$. The matrix 
\begin{equation*}
(A)_{n_1,\dots,n_k\atop\nu_1,\dots,\nu_k}=\langle n_1,\dots,n_k|A|\nu_1,\dots,\nu_k\rangle=a_{n_1,\dots,n_k\atop\nu_1,\dots,\nu_k},
\end{equation*}
is called $k$-unitary if it is unitary for all possible $\binom{2k}{k}$ reordering of its indices, corresponding to all possible choices of $k$ indices out of $2k$. Matrices $k$-unitary for $k>1$ are called \emph{multiunitary} \cite{GALRZ15}. Furthermore, multiunitary matrices are one-to-one connected with \emph{perfect tensors} \cite{PYHP15}, which play an important role in construction of
 holographic codes. 

For instance, a matrix $A$ is 2-unitary if $A$, $A^{T_2}$ and $A^{R}$ are unitary, where $T_2$ and $R$ stand for partial transposition and reshuffling operations, respectively --
see Appendix 2 in Ref. \cite{GALRZ15}.
 As a remarkable property, a matrix $A$ is $k$-unitary if and only 
if the state (\ref{phia}) is $\AME(2k,d)$.

A multiunitary matrix $A$ of size $d^k$ allows us to write a multipartite pure state as the action of a non-local gate acting on $k$ parties over a generalized Bell like-state, that is
\begin{eqnarray}
\label{phia2}
\hspace{-0.4cm}|\phi\rangle\hspace{-0.1cm}&=&\hspace{-0.3cm}\sum_{n_1,\dots,n_k}\hspace{-0.2cm}\bigl(\mathbb{I}_{d^k}\otimes A\bigr)\,|n_1,\dots,n_k\rangle|n_1,\dots,n_k\rangle\nonumber\\
&=&\hspace{-0.2cm}\bigl(\mathbb{I}_{d^k}\otimes A\bigr)\hspace{-0.3cm}\sum_{n_1,\dots,n_k}|n_1,\dots,n_k\rangle|n_1,\dots,n_k\rangle.
\end{eqnarray} 
For any AME state $|\phi\rangle$, the operator $A$ is a non-local $k$-unitary gate acting on $k$ parties. Furthermore, if $A$ is a 2-unitary matrix of size $d^2$, 
then the quantum arrangement
\begin{equation}\label{QOAA}
\begin{array}{ccc}
A|0,0\rangle&\dots&A|0,d-1\rangle\\
\vdots&&\vdots\\
A|d-1,0\rangle&\dots&A|d-1,d-1\rangle
\end{array},
\end{equation}
forms a pair of QLS. In particular, $A$ is a  2-unitary permutation matrix if and only if the arrangement (\ref{QOAA}) is a classical MOLS($d$). This implies that a matrix $A$ being 2-unitary but not permutation defines a quantum QOA. Even more, if the QOA is not equivalent under LOCC to a QOA having associated a permutation matrix $A$ then the QOA is essentially quantum. This is the case of the essentially quantum  array $QOA(4,3_C+2_Q,2,2)$,  presented in Eq.(\ref{QOA4522}). 
In general, $A$ is a $k$-unitary permutation matrix if and only if the state $|\phi\rangle$ 
defined by (\ref{phia2}) is an $\AME(2k,d)$  state with a minimal support. 

In the same way, a $3$-unitary matrix of size $d^3$ defines a set of 3 MOQLC.  As an interesting observation, from Eq.(\ref{phia2}) we realize that any $\AME(2k,d)$ state can be associated with a QOA having at least $N_C=k$ classical columns and the minimal possible number of rows $r=d^{k}$, i.e., it always has index unity.\bigskip

Let us generalise these finding in the following proposition.
\begin{prop}\label{prop5}
A a $k$-unitary matrix $A$ of order $d^k$ defines $m=2k$ MOQLH of size $d$ in dimension $k$. Even more, if $A$ is a permutation matrix then the MOQLH are equivalent to a classical set of MOLH.
\end{prop}
\begin{proof} A $k$-unitary matrix $A$ of size $d^k$ defines an AME state composed 
by $N=2k$ subsystems with $d$ levels each. 
Due to Proposition \ref{prop2} such a state defines 
a QOA($d^k,2k,d,k$) and according to Propostion \ref{prop3}
  implies a MOQLH of size $d$ in dimension $k$. The last implication was already proven in Ref. \cite{GALRZ15}.
\end{proof}

In terms of bipartite quantum gates \cite{CGSS05}, the fact that the classical problem of 36 officers
has no solution,
implies that there is no multi-unitary permutation matrix of size $6^2=36$. That is, there is no permutation  matrix $P_{36}$ of order $36$ such that its partial transpose $P_{36}^{T_2}$ and its reshuffling 
$P_{36}^R$ are both unitary (for an explicit definition of $T_2$ and $R$ see Appendix B in Ref. \cite{GALRZ15}). As a generalization to quantum mechanics, there exists a solution of \emph{36 quantum officers of Euler} 
 if and only if a multi-unitary matrix of size $36$ exists. 
Multi-unitary matrices are relevant in quantum information theory as they saturate the 
upper bound of the entangling power \cite{GALRZ15,ZZF00,Za01}. We remark that on one hand Conjecture \ref{conje2} is consistent with earlier observations by Clarisse \emph{et al.} \cite{CGSS05} and by recent numerical investigations \cite{Puch17,Br17,Lak17}. On the other hand, 
existence of AME(4,6) state cannot be excluded
by applying the currently known bounds for AME states \cite{S04,HESG17,H17}, 
so this interesting problem remains still open.

\section{AME states from quantum orthogonal arrays}\label{S6}
As we have seen in Proposition \ref{prop2}, quantum arrays
$\QOA(r,N,d,k)$ imply existence of $k$-uniform states for $N$ qudit systems having $d$ levels each. In this section, we derive $k$-uniform states with maximal possible value $k=\lfloor N/2\rfloor$ for $N=5$ and arbitrary $d\geq2$ from QOA. Those states determine AME states for $5$-qudit systems.

Let us present a simple construction for AME($5,d$) states for every $d\geq2$ derived from QOA. These states are known to exist \cite{R99} but its explicit closed form has not been presented before, as far as we know. We first define the state
\begin{equation}\label{AME3d}
\AME(3,d)=\sum_{i=0}^{d-1}\ket{i,j,i+j},
\end{equation}
which has associated a classical array $\IrOA(d^2,3,d,1)$. Heres and from now on, sums inside kets is understood to be modulo $d$. By considering this state 
and the generalized Bell basis for 2-qudit systems, 
we are going to construct a QOA composed of 5 columns and $d^2$ rows that defines an $\AME(5,d)$ state for every integer $d$. The first three classical columns of the quantum arrangement are induced by the state (\ref{AME3d}), whereas the remaining two essentially quantum columns are given by elements of the Bell basis
\begin{equation}\label{Bellbasis} 
\ket{\phi_{i,j}}= \sum_{l=0}^{d-1} \omega^{il} \ket{l+j,l},
\end{equation}
where $\omega=\e^{2\pi \mathrm{i} / d}$. We are now in position to establish the following result.
\begin{prop}\label{prop6}
The following three existing quantum objects, determined by a collection
of $d^2$ states $|\phi_{i,j}\rangle \in {\cal H}_{d}^{\otimes 2}$
 are equivalent:
\begin{itemize}
\item[(A)] $\QOA(d^2,3_C+2_Q,d,2)$
\begin{equation}\label{QOA5d}
\begin{array}{cccc}
|0\rangle&|0\rangle&|0\rangle&|\phi_{0,0}\rangle\\
|0\rangle&|1\rangle&|1\rangle&|\phi_{0,1}\rangle\\
\vdots&\vdots&\vdots&\vdots\\
|d-1\rangle&|d-1\rangle&|d-2\rangle&|\phi_{d-1,d-1}\rangle\\
\end{array}.
\end{equation}
\item[(B)] Triple of MOQLS of size $d$
\begin{equation}\label{OAdd5d2}
\begin{array}{ccc}
\ket{0}|\phi_{0,0}\rangle&\dots&\ket{d-1}|\phi_{0,d-1}\rangle\\
\vdots&\ddots&\vdots\\
\ket{d-1}|\phi_{d-1,0}\rangle&\dots&\ket{d-2}|\phi_{d-1,d-1}\rangle\\
\end{array}.
\end{equation}
\item[(C)] Quantum state
\begin{equation} \label{AME5d}
\AME(5,d)= \sum_{i,j=0}^{d-1} \ket{i,\, j,\, i+j} \ket{\phi_{i,j}},
\end{equation}
for any integer $d\geq2$. 
\end{itemize}
\end{prop}
\begin{proof}
Proof of \emph{(A)} follows from two facts: \emph{(i)} every subset of two columns produces an orthonormal basis \emph{(ii)} every reduction to three columns contains orthogonal rows. These conditions ensure that every reduction to two columns produces a POVM. These two properties are an extension of the so-called \emph{uniformity} and \emph{irredundancy}, considered to construct $k$-uniform states from classical OA (see Section IV in Ref. \cite{GZ14}). Equivalence between \emph{(A)} and \emph{(C)}
follows directly from Propostion \ref{prop2}, while the last relation between
\emph{(A)} and \emph{(B)} can be obtained by Propostion \ref{prop3}.
\end{proof}
 
For instance, in the case of $d=2$, this construction reduces to QOA (\ref{QOA4522}), MOQLS (\ref{3MOQLS2}) and AME(5,2) state (\ref{AME52}). Note that the QOA (\ref{QOA5d}) has its last two columns entangled, implying that MOQLS (\ref{OAdd5d2}) are necessarily entangled and AME state (\ref{AME5d}) does not have minimal support. This is consistent with the summary of results presented at the end of Section \ref{S4}.

\begin{obs}
QOA allow us to  add a classical column to arrangement (\ref{QOA5d}) in order to define the following 2-uniform states of 6 qudits, i.e., 
\begin{equation}\label{psi6d}
|\psi(6,d)\rangle=\sum_{i,j=0}^{d-1}\ket{i,j,i+j,i+2j}\ket{\phi_{i,j}},
\end{equation}
where $d$ is an odd prime number
and both sums in kets are taken modulo $d$.\medskip \\
When $d$ is a prime power number, it is convenient to use a polynomial representation based on irreducible polynomials. 
In such cases, the $2$-uniform states of $6$ qubits can be written as
\begin{equation}
|\psi(6,d)\rangle=\sum_{i,j=0}^{d-1}\ket{i,j,i+j,i+a_1 j}\ket{\phi_{i,j}}, \nonumber
\end{equation}
where $a_1$ is the first element of the finite set using the polynomial representation for which $a_1 \neq 0,1$.
\end{obs}
Here, note that the classical and quantum parts of the underlying QOA are composed of four and two columns, respectively. It is simple to check that this arrangement is a QOA($d^2,4_C+2_Q,d,2$).\medskip

In the constructions presented above, the key point was to produce a QOA from combining a classical OA and an orthonormal basis composed of generalized Bell states. It is simple to realize that multiplication of quantum columns produce another QOA having a larger number of columns. For example, the QOA (\ref{QOA4522}) can be extended by considering $m$ copies of the quantum part in the following way:
\begin{equation}\label{oa:qoa}
\begin{array}{llll}
1&1&1&\ket{\Phi^+}\hspace{0.1cm}\dots\hspace{0.17cm}\ket{\Phi^+}\vspace{0.03cm}\\
0&0&1&\ket{\Phi^-}\hspace{0.1cm}\dots\hspace{0.17cm}\ket{\Phi^-}\vspace{0.03cm}\\
0&1&0&\ket{\Psi^+}\hspace{0.1cm}\dots\hspace{0.15cm}\ket{\Psi^+}\vspace{0.03cm}\\
1&0&0&\underbrace{\ket{\Psi^-}\hspace{0.1cm}\dots\hspace{0.1cm}\ket{\Psi^-}}_{m}
\end{array},
\end{equation}
which produces a 2-uniform state of $3+2m$ qubit systems. Furthermore, constructions (\ref{AME5d}) and (\ref{psi6d}) can be generalized in the same way. That is, we construct 2-uniform states for an odd number of $N=5+2m$ qudits
\begin{equation*} 
|\psi(5+m,d)\rangle= \sum_{i,j=0}^{d-1} \ket{i,\, j,\, i+j} \underbrace{\ket{\phi_{i,j}}\cdots\ket{\phi_{i,j}}}_{m},
\end{equation*}
and also 2-uniform states for an even number of $N=6+2m$ qudits
\begin{equation*}
|\psi(6+m,d)\rangle=\sum_{i,j=0}^{d-1}\ket{i,j,i+j,i+2j}
\underbrace{\ket{\phi_{i,j}}\cdots\ket{\phi_{i,j}}}_{m},
\end{equation*}
where $d$ is a prime number. As we described in \eqref{psi6d}, when $d$ is a prime power we should consider the set of polynomial representation of the finite sets. For these constructions it is straightforward to check that every reduction to two parties forms a POVM.\medskip

We recently learned that QOA composed by six columns exist for any prime number of levels $d$. By using qudit graph states \cite{H13}, the following solution can be found \cite{HW17}
for any prime number of levels $d$:
\begin{equation}\label{psi6d}
|AME(6,d)\rangle=\sum_{i_1,i_2,i_3=0}^{d-1}\ket{i_1,i_2,i_3}\ket{\phi_{i_1,i_2,i_3}},
\end{equation}
where 
\begin{equation}
\ket{\phi_{i_1,i_2,i_3}}=\sum_{i_4,i_5,i_6=0}^{d-1}\omega^{A_{i_1,\dots,i_6}}|i_4,i_5,i_6\rangle,
\end{equation}
with $\omega=e^{2\pi i/d}$ and
\begin{eqnarray}
A_{i_1,\dots,i_6}&=&i_1i_2+i_2i_3+i_3i_4+i_4i_5+i_5i_6+i_6i_1\nonumber\\
&&+i_1i_3+i_4i_6+i_2i_5.
\end{eqnarray}
Note that these states determine the $d^3$ entries of three MOQLC of a prime size $d$. Furthermore, these states also imply the existence of a 3-unitary complex Hadamard matrix of size $d^3$ whose entries are given by $M_{\mu,\nu}=\omega^{A_{\mu,\nu}}$, where $\mu=d^2i_1+d i_2+i_3$ and $\nu=d^2i_4+d i_5+i_6$, with $\mu,\nu=0\dots d^3-1$.

\begin{figure*}
\centering 
{\includegraphics[width=14cm]{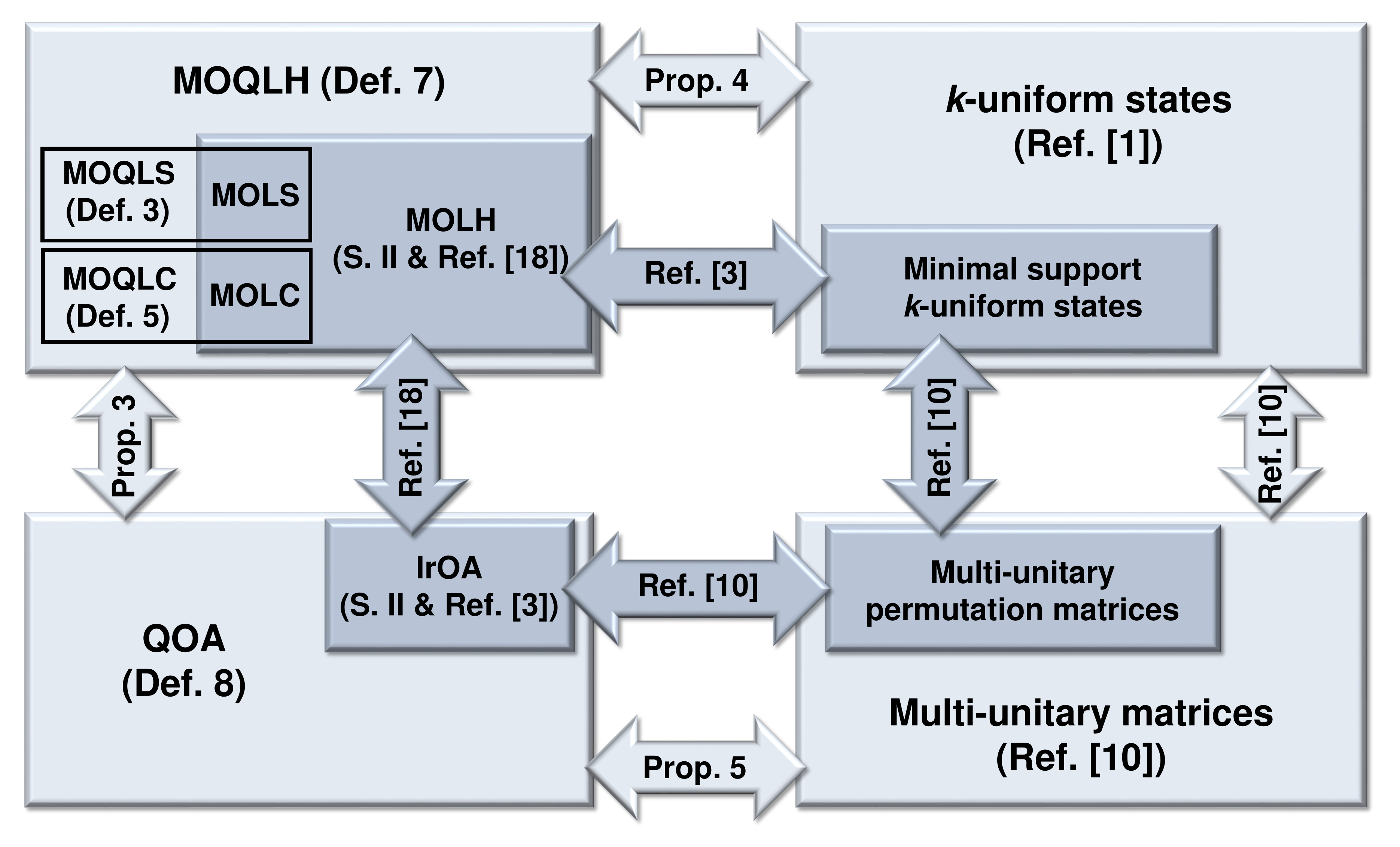}}
\caption{Connection between the quantum combinatorial arrangements introduced in this work, i.e. mutually orthogonal quantum Latin arrangements (see Section \ref{S3}) and quantum orthogonal arrays (see Section \ref{S4}). This also shows the connections existing between our findings, $k$-uniform states and multi-unitary matrices.}
\label{Fig3}
\end{figure*}

\section{Summary and conclusions}\label{S7}
A generalization of classical combinatorial arrangements to quantum mechanics has been established. We introduced the notion of quantum Latin squares (QLS), quantum Latin cubes (QLC), quantum Latin hypercubes (QLH) and established a suitable notion of orthogonality between them -- see Section \ref{S3}. We also introduced 
the notion of quantum orthogonal arrays (QOA) in Section \ref{S4}, that generalizes all the classical and quantum arrangements studied in Sections \ref{S2} and \ref{S3}. Moreover, we derived quantum Latin arrangements from QOA in the same way as classical Latin arrangements can be obtained 
from classical OA -- see Prop. \ref{prop3}.

Our findings allowed us to realize that a pair of orthogonal quantum Latin arrangements not necessarily
 implies existence of two separated arrangements satisfying an orthogonality criteria. 
Indeed, \emph{orthogonal Latin arrangements can be entangled in the same way as quantum states are entangled} --
see for instance Eqs.(\ref{3MOQLS2}) and (\ref{cube}). This astonishing property is one-to-one related to the fact that columns of QOA can be entangled -- see Eqs.(\ref{QOA4522}) and (\ref{QOA8623}). This turned out to be a crucial property in order to reproduce some classes of highly entangled multipartite states, so-called AME states with non-minimal support, for instance the states AME(5,2) and AME(6,2) consisting of five and six qubits, see Eqs. (\ref{AME52}) and (\ref{AME62}), respectively.

Furthermore, QOA define $k$-uniform states -- see Proposition \ref{prop2}. We demonstrated that $k$-uniform states constructed from quantum Latin arrangements have high persistency of entanglement, which makes them ideal candidates for quantum information protocols -- see Prop. \ref{prop4}. We also established a link between multi-unitary matrices and mutually orthogonal Latin arrangements, see Prop. \ref{prop5}. 

We constructed three genuinely entangled MOQLS of size $d$, QOA composed of five columns and an arbitrary number $d$ of internal levels and AME states for five parties with $d$ levels each, for every $d\geq2$ -- see Proposition \ref{prop6}. This result evidences the usefulness of the quantum combinatorial designs introduced along the work.  

Fig. \ref{Fig3} summarizes the relations existing between the studied concepts and the most relevant results derived along the work.
On one hand, we proposed new mathematical tools and described original techniques to construct
multipartite quantum states with remarkable properties. On the other hand,
we established some further links between problems and objects studied in classical 
combinatorics and quantum theory.
We are tempted to believe that such an approach might be fruitful in future
as it  can lead to further development of 'quantum combinatorics' --
a branch of mathematics  which investigates various 
arrangements composed of elements of the continuous
and connected space of $d$-dimensional quantum states instead 
of elements of a discrete set containing $d$ elements.

\section*{Acknowledgements}
It is a pleasure to thank Arul Lakshminarayan, Zbigniew Pucha{\l}a and Wojciech Bruzda
for letting us know about their attempts to find $2$--unitary matrices of size $36$. 
We thank also to Antonio Ac\'in, Christian Gogolin, Markus Grassl, Felix Huber, 
Markus Johansson, Felix Pollock, Arnau Riera, Jens Siewert and Nikolai Wyderka for valuable discussion on quantum orthogonal arrays, AME states and bi-partite unitary 
quantum gates that maximize the entangling power. We are also grateful to the kind hospitality of \emph{Centro de Ciencias Pedro Pascual}, Benasque, Spain, during the \emph{Workshop on Multipartite Entanglement}, May 2016, where this project started.
D. G. and K. \.Z. are supported by the Narodowe Centrum Nauki under the project number DEC- 2015/18/A/ST2/00274 (Poland)
and by the John Templeton Foundation under the project No. 56033.
Z. R. acknowledges support from the Spanish MINECO (QIBEQI FIS2016-80773-P and Severo Ochoa SEV-2015-0522), Fundacio Cellex, Generalitat de Catalunya (CERCA Program), and ERC CoG QITBOX.
S. D. M. is supported by the ERC (Advanced Grant IRQUAT, project number ERC-267386), Spanish MINECO FIS2013-40627-P and FIS2016-80681-P (AEI/FEDER, UE) and Generalitat de Catalunya CIRIT 2014-SGR-966.

\end{document}